\documentclass[12pt]{article}

\usepackage{verbatim}
\usepackage{setspace}
\usepackage{sectsty}
\subsectionfont{\normalsize}
\makeatletter
\renewcommand\section{\@startsection{section}{1}{\z@}%
                                  {-2.0ex \@plus -1ex \@minus -.2ex}%
                                  {2.0ex \@plus.2ex}%
                                  {\normalfont\normalsize\bfseries}}
\makeatother
\usepackage{fullpage}
\usepackage{url}
\usepackage{amsmath}
\usepackage{mathdots}
\newcommand{\DFA}{{\rm DFA }}
\newcommand{\NDFA}{{\rm NDFA }}

\newcommand{\CFG}{{\rm CFG }}

\newcommand{\CSG}{{\rm CSG }}
\newcommand{\DPDA}{{\rm DPDA }}
\newcommand{\PDA}{{\rm PDA }}
\newcommand{\LBA}{{\rm LBA }}

\newcommand{\DTIME}{{\rm DTIME } }

\newcommand{\NSPACE}{{\rm NSPACE } }
\newcommand{\Kleenestar}{{\textstyle *}}

\newcommand{\ang}[1]{\langle#1\rangle}

\newcommand{\goes}{\rightarrow}

\newcommand{\kstar}{{\textstyle *}}

\newcommand{\nat}{{\sf N}}

\newcommand{\xvec}[1]{\ifcase 3{#1} {\ang {x_1,x_2,x_3} } \else 
\ifcase 4{#1} {\ang{x_1,x_2,x_3,x_4}} \else {\ang {x_1,\ldots,x_{#1}}}\fi\fi}
\newcommand{\yvec}[1]{\ifcase 3{#1} {\ang {y_1,y_2,y_3} } \else 
\ifcase 4{#1} {\ang{y_1,y_2,y_3,y_4}} \else {\ang {y_1,\ldots,y_{#1}}}\fi\fi}
\newcommand{\zvec}[1]{\ifcase 3{#1} {\ang {z_1,z_2,z_3} } \else 
\ifcase 4{#1} {\ang{z_1,z_2,z_3,z_4}} \else {\ang {z_1,\ldots,z_{#1}}}\fi\fi}
\newcommand{\vecc}[2]{\ifcase 3{#2} {\ang { {#1}_1,{#1}_2,{#1}_3 } } \else
\ifcase 4{#1} {\ang { {#1}_1,{#1}_2,{#1}_3,{#1}_{4} } }
\else {\ang { {#1}_1,\ldots,{#1}_{#2}}}\fi\fi}
\newcommand{\veccd}[3]{\ifcase 3{#2} {\ang { {#1}_{{#3}1},{#1}_{{#3}2},{#1}_{{#3}3} } } \else
\ifcase 4{#1} {\ang { {#1}_{{#3}1},{#1}_{{#3}2},{#1}_{#3}3},{#1}_{{#3}4} }
\else {\ang { {#1}_{{#3}1},\ldots,{#1}_{{#3}{#2}}}}\fi\fi}
\newcommand{\veccz}[2]{\ifcase 3{#2} {\ang { {#1}_0,{#1}_2,{#1}_3 } } \else
\ifcase 4{#1} {\ang { {#1}_0,{#1}_2,{#1}_3,{#1}_{4} } }
\else {\ang { {#1}_0,\ldots,{#1}_{#2}}}\fi\fi}
\newcommand{\xve}[1]{\ifcase 3{#1} {x_1,x_2,x_3} \else 
\ifcase 4{#1} {x_1,x_2,x_3,x_4} \else {x_1,\ldots,x_{#1}}\fi\fi}
\newcommand{\yve}[1]{\ifcase 3{#1} {y_1,y_2,y_3} \else 
\ifcase 4{#1} {y_1,y_2,y_3,y_4} \else {y_1,\ldots,y_{#1}}\fi\fi}
\newcommand{\zve}[1]{\ifcase 3{#1} {z_1,z_2,z_3} \else 
\ifcase 4{#1} {z_1,z_2,z_3,z_4} \else {z_1,\ldots,z_{#1}}\fi\fi}
\newcommand{\ve}[2]{\ifcase 3#2 {{#1}_1,{#1}_2,{#1}_3} \else
\ifcase 4#2 {{#1}_1,{#1}_2,{#1}_3,{#1}_{4}}
\else {{#1}_1,\ldots,{#1}_{#2}}\fi\fi}
\newcommand{\ved}[3]{\ifcase 3#2 {{#1}_{{#3}1},{#1}_{{#3}2},{#1}_{{#3}3}} \else
\ifcase 4#2 {{#1}_{{#3}1},{#1}_{{#3}2},{#1}_{{#3}3},{#1}_{{#3}4}}
\else {{#1}_{{#3}1},\ldots,{#1}_{{#3}{#2}}}\fi\fi}
\newcommand{\fuve}[3]{
\ifcase 3#2
{{#3}({#1}_1),{#3}({#1}_2,{#3}({#1}_3)} \else
\ifcase 4#2
{{#3}({#1}_1),{#3}({#1}_2),{#3}({#1}_3),{#3}({#1}_4)}
\else
{{#3}({#1}_1),\ldots,{#3}({#1}_{#2})}\fi\fi}
\newcommand{\setmathchar}[1]{\ifmmode#1\else$#1$\fi}
\newcommand{\vlist}[2]{%
	\setmathchar{%
		\compound#2\one{#2}\two
		\ifcompound
			({#1}_1,\ldots,{#1}_{#2})
		\else
			\ifcat N#2
				({#1}_1,\ldots,{#1}_{#2})
			\else
				\ifcase#2
					({#1}_0)\or
					({#1}_1)\or
					({#1}_1,{#1}_2)\or 
					({#1}_1,{#1}_2,{#1}_3)\or
					({#1}_1,{#1}_2,{#1}_3,{#1}_4)\else 
					({#1}_1,\ldots,{#1}_{#2})
				\fi
			\fi
		\fi}}

\newif\ifcompound
\def\compound#1\one#2\two{%
	\def\one{#1}
	\def\two{#2}
	\if\one\two
		\compoundfalse
	\else
		\compoundtrue
	\fi}

\newcommand{\xwe}[1]{\ifcase 3{#1} {x_1\wedge x_2\wedge x_3} \else 
\ifcase 4{#1} {x_1\wedge x_2\wedge x_3\wedge x_4} \else {x_1\wedge \cdots \wedge
x_{#1}}\fi\fi}
\newcommand{\we}[2]{\ifcase 3#2 {\ang { {#1}_1\wedge {#1}_2\wedge {#1}_3 } } \else
\ifcase 4{#1} {\ang { {#1}_1\wedge {#1}_2\wedge {#1}_3\wedge {#1}_{4} } }
\else {\ang { {#1}_1\wedge \cdots\wedge {#1}_{#2}}}\fi\fi}
\newcommand{\st}{\mathrel{:}}

\newcommand{\es}{\emptyset}

\newcommand{\TLE}{\le_{\rm T}}

\newcommand{\TL}{<_{\rm T}}

\newcommand{\s}[1]{\s_{#1}}

\newcommand{\monus}{\;\raise.5ex\hbox{{${\buildrel
    \ldotp\over{\hbox to 6pt{\hrulefill}}}$}}\;}

\newcommand{\infinity}{\infty}

\newcounter{savenumi}

\newtheorem{theoremfoo}{Theorem}[section] 
\newenvironment{theorem}{\pagebreak[1]\begin{theoremfoo}}{\end{theoremfoo}}

\newtheorem{lemmafoo}[theoremfoo]{Lemma}
\newenvironment{lemma}{\pagebreak[1]\begin{lemmafoo}}{\end{lemmafoo}}
\newtheorem{conjecturefoo}[theoremfoo]{Conjecture}

\newtheorem{conventionfoo}[theoremfoo]{Convention}
\newenvironment{convention}{\pagebreak[1]\begin{conventionfoo}\rm}{\end{conventionfoo}}

\newtheorem{porismfoo}[theoremfoo]{Porism}

\newtheorem{gamefoo}[theoremfoo]{Game}

\newtheorem{corollaryfoo}[theoremfoo]{Corollary}
\newenvironment{corollary}{\pagebreak[1]\begin{corollaryfoo}}{\end{corollaryfoo}}

\newtheorem{openfoo}[theoremfoo]{Open Problem}

\newtheorem{exercisefoo}{Exercise}

\newcommand{\fig}[1] 
{
 \begin{figure}
 \begin{center}
 \input{#1}
 \end{center}
 \end{figure}
}

\newtheorem{potanafoo}[theoremfoo]{Potential Analogue}

\newtheorem{notefoo}[theoremfoo]{Note}
\newenvironment{note}{\pagebreak[1]\begin{notefoo}\rm}{\end{notefoo}}

\newtheorem{notabenefoo}[theoremfoo]{Nota Bene}

\newtheorem{nttn}[theoremfoo]{Notation}
\newenvironment{notation}{\pagebreak[1]\begin{nttn}\rm}{\end{nttn}}

\newtheorem{empttn}[theoremfoo]{Empirical Note}

\newtheorem{examfoo}[theoremfoo]{Example}
\newenvironment{example}{\pagebreak[1]\begin{examfoo}\rm}{\end{examfoo}}

\newtheorem{dfntn}[theoremfoo]{Def}
\newenvironment{definition}{\pagebreak[1]\begin{dfntn}\rm}{\end{dfntn}}

\newtheorem{propositionfoo}[theoremfoo]{Proposition}

\newenvironment{proof}
    {\pagebreak[1]{\narrower\noindent {\bf Proof:\quad\nopagebreak}}}{\QED}

\newcommand{\yyskip}{\penalty-50\vskip 5pt plus 3pt minus 2pt}
\newcommand{\blackslug}{\hbox{\hskip 1pt
        \vrule width 4pt height 8pt depth 1.5pt\hskip 1pt}}
\newcommand{\QED}{{\penalty10000\parindent 0pt\penalty10000
        \hskip 8 pt\nolinebreak\blackslug\hfill\lower 8.5pt\null}
        \par\yyskip\pagebreak[1]}

\newcommand{\BBB}{{\penalty10000\parindent 0pt\penalty10000
        \hskip 8 pt\nolinebreak\hbox{\ }\hfill\lower 8.5pt\null}
        \par\yyskip\pagebreak[1]}

\newtheorem{factfoo}[theoremfoo]{Fact}
\newenvironment{fact}{\pagebreak[1]\begin{factfoo}}{\end{factfoo}}

\newenvironment{block}{\begin{list}{\hbox{}}{\leftmargin 1em
    \itemindent -1em \topsep 0pt \itemsep 0pt \partopsep 0pt}}{\end{list}}

\begin{document}

\newcommand{\KN}{K_{\nat}}
\newcommand{\NRE}{\hbox{NUM-RED-EDGES\ }}
\newcommand{\NBE}{\hbox{NUM-BLUE-EDGES\ }}
\newcommand{\RED}{\hbox{RED\ }}
\newcommand{\BLUE}{\hbox{BLUE\ }}
\newcommand{\REDns}{\hbox{RED}}
\newcommand{\BLUEns}{\hbox{BLUE}}

\centerline{\bf On the Sizes of DPDAs, PDAs, LBAs}
\centerline{by Richard Beigel and William Gasarch}

\begin{abstract}
There are languages $A$ such that there is a Pushdown Automata (PDA) that
recognizes $A$ which is much smaller 
than any Deterministic Pushdown Automata (DPDA) 
that recognizes $A$.
There are languages $A$ such that there is a Linear Bounded Automata
(Linear Space Turing Machine, henceforth LBA) that
recognizes $A$ which is much smaller than any PDA that recognizes $A$.
There are languages $A$ such that both $A$ and $\overline{A}$ are recognizable
by a PDA, but the PDA for $A$ is much smaller than the PDA for $\overline{A}$.
There are languages $A_1,A_2$ such that $A_1,A_2,A_1\cap A_2$ are recognizable
by a PDA, but the PDA for $A_1$ and $A_2$ are much smaller than the PDA for $A_1\cap A_2$.
We investigate 
these phenomena and show that, in all these cases,
the size difference is captured by a function whose
Turing degree is on the second level of the arithmetic hierarchy.

Our theorems lead to 
infinitely-many-$n$ results. For example: for-infinitely-many-$n$ there exists
a language $A_n$ recognized by a DPDA
such that there is a small PDA for $A_n$, but any DPDA for $A_n$ is very large.
We look at cases where we can get all-but-a-finite-number-of-$n$ results, though with much
smaller size differences.
\end{abstract}


\section{Introduction}

Let DPDA be the set of Deterministic Push Down Automata,
PDA be the set of Push Down Automata, and
LBA be the set of Linear Bounded Automata (usually called nondeterministic
linear-space bounded Turing Machines).
Let $L(\DPDA)$ be the set of languages recognized
by DPDAs (similar for $L(\PDA)$ and $L(\LBA)$).
It is well known that

$$L(\DPDA) \subset L(\PDA)\subset L(\LBA).$$

Our concern is with the {\it size} of the DPDA, PDA, LBA.
For example, let $A\in L(\DPDA)$. Is it possible that there is a PDA for $A$
that is much smaller than any DPDA for $A$? For all adjacent pairs above we
will consider these questions. 
There have been related results by
Valiant~\cite{valdpda},
Schmidt~\cite{Schmidt},
Meyer and Fischer~\cite{ecodesc},
Hartmanis~\cite{hartlang}, and Hay~\cite{haylang}.
We give more details on their results later.

Throughout the paper $\Sigma$ is a finite alphabet and \$
is a symbol that is not in $\Sigma$. All of our languages will either
be subsets of $\Sigma^\kstar$ or $(\Sigma\cup \{\$\})^\kstar$.

\begin{convention}
A {\it device} will either be a recognizer (e.g., a DFA) or
a generator (e.g., a regular expression).
We will use $\cal M$ to denote a set of devices (e.g., DFAs).
We will refer to an element of $\cal M$ as an $\cal M$-device.
If $P$ is an $\cal M$-device then let $L(P)$ be the language recognized
or generated by $P$.
Let $L({\cal M})=\{L(P) \st P\in {\cal M}\}$.
\end{convention}

\begin{definition}\label{de:devices}
Let $\cal M$ and $\cal M'$ be two sets of devices
such that $L({\cal M})\subseteq L({\cal M'})$.
(e.g., DFAs and DPDAs).
A {\it bounding function for $({\cal M},{\cal M'})$} is a function $f$ such that
for all $A \in L({\cal M})$, if $A\in L({\cal M'})$ via a device of 
size $n$
then 
$A \in L({\cal M})$ via a device of size $\le f(n)$.
\end{definition}

\begin{definition}~
\begin{enumerate}
\item
The {\it size of a DFA or NDFA} is its number of states.
\item
The {\it size of a DPDA or PDA} is the sum of its number of states and
its number of symbols in the stack alphabet.
\item
The {\it size of a CFG or CSL} is its number of nonterminals.
\item
The {\it size of an LBA} is the sum of its number of states and 
its number of symbols in the alphabet (note that the alphabet used
by the Turing machine may be bigger than the input alphabet).
\end{enumerate}
\end{definition}

We now give some examples and known results.

\begin{example}\label{ex:bd}~
{\bf Known Upper Bounds:}
\begin{enumerate}
\item
$f(n)=2^n$ is a bounding function for (DFA,NDFA)
by the standard proof that $L(\NDFA)\subseteq L(\DFA)$.
\item
$f(n)=n^{{n^{n^{O(n)}}}}$ is a bounding function for (DFA,DPDA).
This is a sophisticated construction by Stearns~\cite{regpda}.
\item
$f(n)=2^{{2^{O(n)}}}$ is a bounding function for (DFA,DPDA).
This is a sophisticated construction by Valiant~\cite{valreg}.
Note that this is a strict improvement over the construction of Stearns.
\item
$f(n)=O(n^{O(1)})$ is a bounding function for (CFG,PDA).
This can be obtained by an inspection of the proof that
$L(\PDA)\subseteq L(\CFG)$.
\item
$f(n)=O(n)$ is a bounding function for (PDA,CFG).
This can be obtained by an inspection of the proof that
$L(\CFG)\subseteq L(\PDA)$.
\item
$f(n)=O(n)$ is a bounding function for (CSG,LBA).
This can be obtained by an inspection of the proof that
$L(\LBA)\subseteq L(\CSG)$.
\item
$f(n)=O(n)$ is a bounding function for (LBA,CSG).
This can be obtained by an inspection of the proof that
$L(\CSG)\subseteq L(\LBA)$.
\end{enumerate}
\end{example}

\noindent
\begin{example}\label{ex:lower}
{\bf Known Lower Bounds:}
\begin{enumerate}
\item
Meyer and Fischer~\cite{ecodesc} 
proved that 
(1) If $f$ is the bounding function for (DFA,NDFA) then $2^n\le f(n)$.
(2) If $f$ is the bounding function for (DFA,DPDA) then $2^{2^{O(n)}}\le f(n)$.
(3) If $f$ is the bounding function for (DFA,CFG) then $HALT\TLE f$.
The sets they used for (3) were finite.
\item
Let UCFG be the set all unambiguous context free grammars.
Valiant~\cite{valdpda} showed that 
if $f$ is the bounding function for (DPDA,UCFG) then $HALT\TLE f$.
\item
Schmidt~\cite{Schmidt}
showed that 
if $f$ is the bounding function for (UCFG,CFG) then $HALT\TLE f$.
\item
Hartmanis~\cite{hartlang} showed that
if $f$ is the bounding function for (DPDA,PDA) then $HALT\TLE f$.
\item
Harel and Hirst (\cite{hhpda}, 
Proposition 14 and Corollary 15 of the journal version,
(Proposition 12 and Corollary 13 of the conference version)
have shown that if $g$ is computable then
$g$ is not a bounding function for (DPDA,PDA) in a very strong way.
They showed that {\it for all $n>0$} there is a language $L_n$ 
such that (1) there is a PDA for $L_n$ of size $O(n)$, but
(2) any DPDA for $L_n$ requires size at least $g(n)$.
\item
Hay~\cite{haylang} showed that
if $f$ is the bounding function for (DPDA,PDA) then $HALT\TL f$.
She also showed that there is a bounding function $f$ for (DPDA,PDA) such
that $f\TLE INF$.
($INF$ is the set of all indices of Turing machines that halt on an infinite number
of inputs. It is complete for the second level of the arithmetic hierarchy and
hence strictly harder than $HALT$.)
\item
Gruber et al.~\cite{descabstract} proved several general theorems about
sizes of languages. Hay's result above is a corollary of their
theorem.
\end{enumerate}
\end{example}

\begin{note}
The results above that conclude $HALT\TLE f$ were not stated that way
in the original papers. They were stated as either {\it $f$ is not recursive}
or {\it $f$ is not recursively bounded}. However, an inspection of their proofs
yields that they actually proved $HALT\TLE f$.
\end{note}

\begin{definition}\label{de:cbd}
Let $\cal M$  be a set of devices.
A {\it c-bounding function for $\cal M$} is a function $f$ such that
for all $A$ that are recognized by an $\cal M$-device of size $n$,
if $\overline{A}\in L({\cal M})$ then it is  recognized by an $\cal M$- device,
of size $\le f(n)$.
One linguistic issue--- we will write (for example)
{\it c-bounding function for PDAs}
rather than 
{\it c-bounding function for PDA} since the former flows better verbally.
\end{definition}

We now give some examples and known results.

\begin{example}\label{ex:cbd}~
\begin{enumerate}
\item
$f(n)=2^n$ is a c-bounding function for NDFAs.
This uses the standard proofs that $L(\NDFA)\subseteq L(\DFA)$ and
that $L(\DFA)$ is closed under complementation.
\item
$f(n)=O(n)$ is a c-bounding function for DPDAs. This is an easy exercise in
formal language theory.
\item
$f(n)=O(n)$ is a c-bounding function for LBAs. This can be obtained by an inspection of
the proof, by Immerman-Szelepcsenyi~\cite{Immerman,Sz}, that nondeterministic linear space is closed under complementation.
\end{enumerate}
\end{example}

\begin{definition}\label{de:ibd}
Let $\cal M$  be a set of devices.
An {\it i-bounding function for $\cal M$} is a function $f$ such that
for all $A_1,A_2$ that are recognized by an $\cal M$-device of size $n$,
if $A_1\cap A_2\in L({\cal M})$ then it is  recognized by an $\cal M$- device,
of size $\le f(n)$.
One linguistic issue--- we will write (for example)
{\it i-bounding function for PDAs}
rather than 
{\it i-bounding function for PDA} since the former flows better verbally.
\end{definition}

\begin{example}\label{ex:ibd}~
\begin{enumerate}
\item
$f(n)=2n$ is an i-bounding function for DFA.
This uses the standard proofs that $L(\DFA)$ is closed under intersection.
\item
$f(n)=2^{2n}$ is an i-bounding function for NDFAs.
Convert both NDFAs to DFAs and then use the standard proof that $L(\DFA)$
is closed under intersection.
\end{enumerate}
\end{example}

\begin{note}
We will state our results in terms of DPDAs, PDAs, and LBAs.
Hence you may read expressions like $L(\PDA)$ and think 
{\it isn't that just CFLs?} It is. We do this to
cut down on the number of terms this paper refers to.
\end{note}

\section{Facts and Notation}

We will need the following notation and facts to state our results.
We will prove the last item since it seems to not be as well known as
the others.

\begin{fact}\label{fa:turing}~
\begin{enumerate}
\item
$M_0,M_1,M_2,\ldots$ is a standard numbering of all deterministic Turing Machines.
\item
$M_{e,s}(x)$ is the result of running $M_e(x$) for $s$ steps.
\item
$HALT$ is the set $\{ (e,x)\st (\exists s)[M_{e,s}(x)\hbox{ halts }\}$. 
$HALT$ is $\Sigma_1$-complete. 
Hence any $\exists$ question can be phrased as a query to $HALT$.
Note that any $(\forall)$ question can also be phrased as a query; however,
you will have to negate the answer.
\item
$INF$ is the set $\{ e\st (\forall x)(\exists y,s)[y>x \wedge M_{e,s}(y)\hbox{ halts }]\}$.
$INF$ is $\Pi_2$-complete. 
Hence any $(\forall)(\exists)$ question can be phrased as a query to $INF$.
Note that any $(\exists)(\forall)$ question can also be phrased as a query; however,
you will have to negate the answer.
\item
$A\TLE B$ means that $A$ is decidable given complete access to set $B$.
This can be defined formally with oracle Turing machines.
\item
$f\TLE HALT$ iff there exists a computable $g$ such that, for all $n$,
$f(n)=\lim_{s\goes\infinity} g(n,s)$.
We can take $g$ to have complexity $O(\log(n+s))$, or even lower.
This result is due to Shoenfield~\cite{Shoenfield} and is referred to as
{\it The Shoenfield Limit Lemma.} It is in most computability theory books.
Note that the domain $f$ is $\nat$, the domain of $g$ is $\nat\times\nat$,
and the codomain of both $f$ and $g$ is $\nat$. Hence the $\lim_{s\goes\infinity} g(n,s)$
means that $(\exists s_0,x)(\forall s\ge s_0)[g(n,s)=x]$.
\end{enumerate}
\end{fact}

\begin{proof}
We just prove Part 6.

Let
$HALT_s=\{ (e,x)\st [M_{e,s}(x)\hbox{ halts }\}$.

Assume $f\TLE HALT$ via oracle Turing machine $M_i^{()}$.
Hence, for all $n$, $M_i^{HALT}(n)$ halts and is equal to $f(n)$.
Since $M^{HALT}(n)$ halts it does so in a finite amount of time and using
only a finite number of queries to $HALT$. Hence there exists $s$ such that
for all $t\ge s$, $f(n) = M_{i,t}^{HALT_t}(n)$.
Therefore the following computable function $g$ works.
\begin{equation}
g(n,s) =
\begin{cases}
0                     & \text{ if $M_{i,s}^{HALT_s}(n)$ has not converged } \\
\hbox{ $M_{i,s}^{HALT_s}(n)$}  & \text{ if it has converged } 
\end{cases}
\end{equation}

We can obtain a $g$ of very low complexity by 
(for example) having $g(n,s)$ 
computer
$M_{i,\lg^* s}^{HALT_{\lg^* s}}$ 
instead of
$M_{i,s}^{HALT_{s}}$ 

Assume there exists a computable $g$ such that, for all $n$,
$f(n)=\lim_{s\goes\infinity} g(n,s)$.
The following algorithm shows $f\TLE HALT$.

\noindent
ALGORITHM

\noindent
Input($n$)

\noindent
For $s_0=1$ to $\infinity$ (we will show that this terminates)

\qquad Ask $HALT$ $(\forall s\ge s_0)[ g(n,s)=g(n,s+1)]$.

\qquad If YES then output $g(n,s)$ and STOP.

\noindent
END OF ALGORITHM

Since $\lim_{s\goes\infinity} g(n,s)$ exists there will be an $s_0$ such that
the answer to the $HALT$ question is YES. Hence the algorithm terminates.

\end{proof}

\section{Summary of Results}\label{se:summary}

In this section we summarize our results. We also present them
in a table at the end of this section.

The results of Hartmanis~\cite{hartlang} and Hay~\cite{haylang} mentioned
in Exercise~\ref{ex:lower}
above leave open the exact Turing degree of the bounding function for
(DPDA,PDA).
In Section~\ref{se:gen} we resolve this question by proving a 
general theorem from which we obtain the following:

\begin{enumerate}
\item
If $f$ is a bounding function for (DPDA,PDA) then $INF\TLE f$.
\item
There exists a bounding function for (DPDA,PDA) such that $f\TLE INF$.
(Hay~\cite{haylang} essentially proved this; however, we restate and reprove in
our terms.)
\item
If $INF\not\TLE f$ then for infinitely many $n$ there exists a language 
$A_n\in L(\DPDA)$
such that 
(1) any DPDA that recognizes $A_n$ requires size $\ge f(n)$,
(2) there is a PDA of size $\le n$ that recognizes $A_n$.
(This follows from Part 1.)
\item
If $f$ is a bounding function for (PDA,LBA) then $INF\TLE f$.
\item
There exists a bounding function for (PDA,LBA) such that $f\TLE INF$.
\item
If $INF\not\TLE f$ then for infinitely many $n$ there exists a 
language $A_n\in L(\PDA)$
such that 
(1) any PDA that recognizes $A_n$ requires size $\ge f(n)$,
(2) there is an LBA of size $\le n$ that recognizes $A_n$.
(This follows from Part 4.)
\end{enumerate}

In Section~\ref{se:cbd} and~\ref{se:ibd} we find the exact Turing degree
of the c-bounding function and the i-bounding function for PDAs.
We obtain the following:

\begin{enumerate}
\item
If $f$ is a c-bounding function for PDAs then $INF\TLE f$.
\item
There exists a c-bounding function for PDAs such that $f\TLE INF$.
\item
If $INF\not\TLE f$ then for infinitely many $n$ 
there exists a language $A_n$
such that (1) $A_n,\overline{A_n}\in L(\PDA)$,
(2) there is no PDA of size $\le f(n)$ for $\overline{A_n}$, but
(3) there is a PDA of size $\le n$ for $A_n$.
(This follows from Part 1.)
\item
Results 1,2,3 but with i-bounding functions instead of c-bounding functions.
\end{enumerate}

Note that we have several results of the form 
{\it for infinitely many $n$ $\ldots$} that use unnatural languages.
We would like to have 
{\it for all but finitely many $n$ $\ldots$} results that involve natural languages.
We need the following definitions.

\begin{definition}
Let $A(n)$ be a statement about the natural number $n$.
$A(n)$ {\it is true for almost all $n$} means that $A(n)$ is
true for all but a finite number of $n$.
\end{definition}

\begin{definition}
(Informal)
A language is {\it unnatural} if it exists for the sole point of proving
a theorem. 
\begin{example}~
\begin{enumerate}
\item
Languages that involve Turing configurations are not natural.
\item
Languages created by diagonalization are not natural.
\item
The language $\{ ww\st |w|=n \}$ is natural.
\end{enumerate}
\end{example}
\end{definition}

\begin{note}
We will sometimes state theorems as follows:
{\it there exists a (natural) language such that $\ldots$}.
If we do not state it that way then the language is unnatural.
\end{note}

In Sections~\ref{se:bigpda} 
we obtain the following {\it for almost all $n$} results\footnote{Harel and Hirst\cite{hhpda} obtained these independently 21 years ago. We comment on how their
 proofs and our proofs differ in Section~\ref{se:bigpda}.}.

\noindent
For almost all $n$ there exists a (natural) language $A_n\in L(\DPDA)$ 
such that 
\begin{enumerate}
\item
Any DPDA for $A_n$ requires size $\ge 2^{2^{\Omega(n)}}$.
\item
There is a PDA of size $O(n)$ that recognizes $A_n$.
\end{enumerate}

\begin{note}
As noted in Example~\ref{ex:lower}.5
Harel and Hirst~\cite{hhpda} have a stronger result; however,
the language they uses is not natural.
\end{note}

Section~\ref{se:bigpda} also has the following result. In fact,
we prove this result first and derive the result above from it.

\noindent
For almost all $n$ there exists a (natural) language $A_n$ such that
\begin{enumerate}
\item
$A_n,\overline{A_n}\in L(\PDA)$.
\item
Any PDA for $\overline{A_n}$ requires size $\ge {2^{2^{\Omega(n)}}}$.
\item
There is a PDA of size $O(n)$ that recognizes $A_n$.
\end{enumerate}

%
%

In Section~\ref{se:bigishlba} we show the following:

\noindent
For almost all $n$ there exists a (natural) language $A_n\in L(\PDA)$
such that 
\begin{enumerate}
\item
Any PDA for $A_n$ requires size $\ge 2^{2^{\Omega(n)}}$.
\item
There is an LBA of size $O(n)$ that recognizes $A_n$.
\end{enumerate}

In Section~\ref{se:biglba} we obtain\footnote{Meyer originally claimed this result. See
the discussion in Section~\ref{se:biglba}.} a {\it for almost all $n$} result
for (PDA,LBA):

\noindent
Let $f$ be any function such that $f\TLE HALT$.
For almost all $n$ there exists a finite language $A_n$ such that 
\begin{enumerate}
\item
Any PDA for $A_n$ requires size $\ge f(n)$.
\item
There is an LBA of size $O(n)$ that recognizes $A_n$.
\end{enumerate}

We summarize our results in the following tables.

\noindent
{\bf The first table:}
The first two columns, ${\cal M}$ and ${\cal M'}$,
indicate that we are looking at the bounding functions $f$ for 
$({\cal M},{\cal M'})$. The third column yields a property that
any such $f$ must have.
The fourth column states whether we get as a corollary to the proof a
result about
for-infinitely-many $n$ (io) or for-almost-all-n (ae).
The fifth column states if the sets involved are natural or not.
The sixth column states if the condition on $f$ from the third column
captures exactly the Turing degree of the bounding function.
For example, in the first column we see that any bounding
function $f$ for $(\DPDA,\PDA)$ has to satisfy $INF\TLE f$; however,
the YES in column six indicates that there is a bounding function of this
Turing degree.

\[
\begin{array}{|c|c|c|c|c|c|}
\hline
{\cal M} & {\cal M'} & (\forall f) & \hbox{io/ae} & \hbox{Nat?} & \hbox{Exact TD?} \cr
\hline
\DPDA & \PDA & INF\TLE f              & io & NO & YES \cr
\PDA  & \LBA & INF\TLE f              & io & NO & YES \cr
\DPDA & \PDA & f\ge 2^{2^{\Omega(n)}} & ae & YES & NO \cr
\PDA  & \LBA & f\ge 2^{2^{\Omega(n)}} & ae & YES & NO \cr
\PDA  & \LBA & f\not\TLE HALT         & ae & NO & YES \cr
\hline
\end{array}
\]

\noindent
{\bf The second table:}
The first two columns are a set of devices, $\cal M$,
and an operation Op (either Complementation (c) or Intersection (i)).
We are concerned with the Op-bounding functions $f$ for
$\cal M$.
The third column yields a property that any such $f$ must have.
The fourth column states whether we get as a corollary to the proof a
result about for-infinitely-many $n$ (io) or for-almost-all-n (ae).
The fifth column states if the sets involved are natural or not.
The sixth column states if the condition on $f$ from the third column
captures exactly the Turing degree of the bounding function.

\[
\begin{array}{|c|c|c|c|c|c|}
\hline
{\cal M} & \hbox{Op} & (\forall f)                  & \hbox{io/ae}   & \hbox{Nat?} & \hbox{Exact TD?} \cr
\hline
\PDA & \hbox{Complementation}   & INF\TLE f                    & io             & NO          & YES \cr
\PDA & \hbox{Intersection}  & INF\TLE f                    & io             & NO          & YES \cr
\PDA & \hbox{Complementation}   & f\ge 2^{2^{\Omega(n)}}       & ae             & YES         & NO \cr
\hline
\end{array}
\]

\section{Bounding Functions for (DPDA,PDA) and (PDA,LBA)}\label{se:gen}

In this section we prove a general theorem about bounding
functions and then apply it to both (DPDA,PDA) and (PDA,LBA).
In both cases we show that the Turing degree of the bounding function
is in the second level of the arithmetic hierarchy.

We will need to deal just a bit with actual Turing Machines. 

\begin{definition}
Let $M$ be a deterministic Turing Machine.
A {\it configuration (henceforth config)}  of $M$ is a string
of the form $\alpha\substack{q\\\sigma}\beta$ where $\alpha,\beta\in\Sigma^\Kleenestar$,
$\sigma\in\Sigma$, and
$q\in Q$. We interpret this as saying that the machine 
has $\alpha\sigma\beta$ on the tape (with blanks to the left and right),
is in state $q$, and the head is looking at the
square where we put the $\substack{q\\\sigma}$. 
Note that from the config one can determine if the machine has
halted, and also, if not, what the next config is.
\end{definition}

\begin{notation}
If $C$ is a string then $C^R$ is that string written backwards.
For example, if $C=aaba$ then $C^R=abaa$.
\end{notation}

\begin{definition}\label{de:acc}
Let $e,x\in\nat$.
We assume that any halting computation of $M_e$ takes an even number of steps.
\begin{enumerate}
\item
Let $ACC_{e,x}$ be the set of all sequences of config's represented by

$$
\$C_1
\$C_2^R
\$C_3
\$C_4^R
\$
\cdots
\$
C_s^R
\$
$$

such that
\begin{itemize}
\item
$|C_1|=|C_2| = \cdots = |C_s|$.
\item
The sequence $C_1,C_2,\ldots,C_s$ represents an accepting computation of $M_e(x)$.
\end{itemize}
\item
Let $ACC_e= \bigcup_{x\in\nat} ACC_{e,x}$.
\end{enumerate}
\end{definition}

Hartmanis~\cite{hartlang} proved the following lemma.

\begin{lemma}\label{le:intcfl}
For all $e,x$, $\overline{ACC_{e,x}}\in L(\PDA)$.
For all $e$, $\overline{ACC_{e}}\in L(\PDA)$.
In both cases it is computable to take the parameters ($(e,x)$ or $e$)
and obtain the PDA.
\end{lemma}

\begin{definition}
Let $\cal M$ and $\cal M'$ be two sets of devices.
\begin{enumerate}
\item
${\cal M}\subseteq {\cal M'}$ {\it effectively} if
there is a computable function that will, given an
$\cal M$-device $P$, output an $\cal M'$-device $P'$
such that $L(P)=L(P')$.
\item
${\cal M}$ is {\it effectively closed under complementation}
if
there is a computable function that will, given an
$\cal M$-device $P$, output an $\cal M$-device $P'$
such that $L(P')=\overline{L(P)}$.

\item
The {\it non-emptiness problem for $\cal M$} is the following:
{\it given an $\cal M$-device $P$ determine if $L(P)\ne\es$}.

\item
The {\it membership problem for $\cal M$} is: {\it given an $\cal M$-device $P$
and $x\in\Sigma^\kstar$ determine if $x\in L(P)$.}

\item
${\cal M}$ is {\it size-enumerable} if there exists a list of devices
$P_1,P_2,\ldots$ such that 
\begin{itemize}
\item
${\cal M}=\{L(P_i) \st i\in \nat\}$,
\item
$(\forall i)[|P_i|\le |P_{i+1}|]$, and
\item
the function from $i$ to $P_i$ is computable.
\end{itemize}
Note that DFA, NDFA, DPDA, PDA, LBA are all size-enumerable, however
UCFG is not.
\end{enumerate}
\end{definition}

\begin{theorem}\label{th:bdd}
Let $\cal M$ and $\cal M'$ be two sets of devices
such that the following hold.
\begin{itemize}
\item
$L({\cal M}) \subseteq L(\PDA) \subseteq L({\cal M'})$ effectively.
\item
At least one of $\cal M$,  $\cal M'$ is effectively closed under complementation.
\item
The non-emptiness problem for $\cal M$ is decidable.
\item
The membership problems for $\cal M$ and $\cal M'$  are decidable.
\item
Every finite set is in $L({\cal M})$.
\item
${\cal M}$ is size-enumerable.
\end{itemize}
Then
\begin{enumerate}
\item
If $f$ is a bounding function for $({\cal M},{\cal M'})$ then $HALT\TLE f$.
\item
If $f$ is a bounding function for $({\cal M},{\cal M'})$ then $INF\TLE f$.
\item
There exists a bounding function $f\TLE INF$ for $({\cal M},{\cal M'})$.
\item
If $INF\not\TLE f$ then for infinitely many $n$ there exists a language 
$A_n\in L({\cal M})$
such that 
(1) any $\cal M$-device that recognizes $A_n$ requires size $\ge f(n)$,
(2) there is an $\cal M'$-device of size $\le n$ that recognizes $A$.
(This follows from Part 2 so we do not prove it.)
\end{enumerate}
\end{theorem}

\begin{proof}

\noindent
1) If $f$ is a bounding function for $({\cal M},{\cal M'})$ then $HALT\TLE f$.

Note that

\begin{itemize}
\item
If $M_e(x)$ halts then $ACC_{e,x}$ has one string, which is the accepting computation of $M_e(x)$.
\item
If $M_e(x)$ does not halt  then $ACC_{e,x}=\es$.
\item
Given $e,x$ one can construct a PDA for $\overline{ACC_{e,x}}$ by Lemma~\ref{le:intcfl}.
\end{itemize}

We give the algorithm for $HALT\TLE f$.
There will be two cases in it depending on which of
$\cal M$ or $\cal M'$ is effectively closed under complementation.

\noindent
{\bf ALGORITHM FOR $HALT\TLE f$}

\begin{enumerate}
\item
Input$(e,x)$
\item
Construct the PDA $P$ for $\overline{ACC_{e,x}}$.
Obtain the device $Q$ in $\cal M'$ that accepts $\overline{ACC_{e,x}}$.

\item

\noindent
{\bf Case 1:} $\cal M$ is effectively closed under complementation.
Compute $f(|Q|)$.
Let $D_1,\ldots,D_t$ be all of the $\cal M$-devices of size $\le f(|Q|)$.
Create the $\cal M$ devices for their complements,
which we denote $E_1,\ldots,E_t$.

\noindent
{\bf Case 2:} $\cal M'$ is effectively closed under complementation.
Find an ${\cal M'}$- device $R$ for $\overline{L(Q)}=ACC_{e,x}$.
Compute $f(|R|)$.
Let $E_1,\ldots,E_t$ be all of the $\cal M$-devices of size $\le f(|R|)$.

\bigskip
Note that at the end of step 3, regardless of which case happened,
we have a set of ${\cal M}$-devices $E_1,\ldots,E_t$
such that

$(e,x)\in HALT$ iff 

$(\exists 1\le i\le t)[L(E_i)\hbox{ is one string which represents an accepting computation of $M_e(x)$}]$.

\item
For each $1\le i\le t$ (1) determine if $L(E_i)=\es$
(2) if $L(E_i)=\es$ then let $w_i$ be the empty string, and
(3) if $L(E_i)\ne \es$ then, in lexicographical order, test strings for membership
in $L(E_i)$ until you find a string in $L(E_i)$ which we denote $w_i$.
If $\{w_1,\ldots,w_s\}$ contains a string representing an
accepting computation of $M_e(x)$ then output YES.
If not then output NO.
\end{enumerate}

\noindent
{\bf END OF ALGORITHM}

\bigskip

\noindent
2) If $f$ is a bounding function for $({\cal M},{\cal M'})$ then $INF\TLE f$.

Note that

\begin{itemize}
\item
If $e\in INF$ then $ACC_e\notin L(\PDA)$
since
$ACC_e$ is infinite and every string in it begins with
$\$C_1\$C_2^R\$C_3\$$ where $|C_1|=|C_2^R|=|C_3|$.
\item
If $e\notin INF$ then $ACC_e\in L(\PDA)$ since $ACC_e$ is finite.
\item
Given $e$ one can construct a PDA for $\overline{ACC_{e}}$ by Lemma~\ref{le:intcfl}.
\end{itemize}

We give the algorithm for $INF\TLE f$. 
There will be two cases in it depending on which of 
$\cal M$ or $\cal M'$ is effectively closed under complementation.

In the algorithm below we freely use Fact~\ref{fa:turing}.2 to phrase
$(\exists)$-questions as queries to $HALT$, and Part 1 to answer
queries to $HALT$ with calls to $f$.

\noindent
{\bf ALGORITHM FOR $INF\TLE f$}

\begin{enumerate}
\item
Input$(e)$
\item
Construct the PDA $P$ for $\overline{ACC_{e}}$.
Obtain the device $Q$ in $\cal M'$ that accepts $\overline{ACC_{e}}$.
\item
There are two cases.

\noindent
{\bf Case 1:} $\cal M$ is effectively closed under complementation.
Compute $f(|Q|)$.
Let $D_1,\ldots,D_t$ be all of the $\cal M$-devices of size $\le f(|Q|)$.
Create the $\cal M$ devices for their complements,
which we denote $E_1,\ldots,E_t$.

\noindent
{\bf Case 2:} $\cal M'$ is effectively closed under complementation.
Find an ${\cal M'}$- device $R$ for $\overline{L(Q)}=ACC_e$.
Compute $f(|R|)$.
Let $E_1,\ldots,E_t$ be all of the $\cal M$-devices of size $\le f(|R|)$.

\bigskip

Note that at the end of step 3, regardless of which case happened,
we have a set of ${\cal M}$-devices $E_1,\ldots,E_t$
such that

$e\in INF$ 
$\implies ACC_e\notin L(\PDA)$ 
$\implies$ 
$ACC_e\notin L({\cal M})$ 
$\implies$ 
$ACC_e\notin \{L(E_1),\ldots,L(E_t)\}$ 
$\implies$ 
$(\exists x_1,\ldots,x_t)(\forall 1\le i\le t)[ACC_e(x_i)\ne E_i(x_i)]$.

\bigskip

$e\notin INF$ 
$\implies ACC_e\hbox{ is finite }$
$\implies$
$ACC_e\in L({\cal M})$
$\implies$
$(\exists 1\le i\le t)[L(E_i)=ACC_e]$
$\implies$
$\neg(\exists x_1,\ldots,x_t)(\forall 1\le i\le t)[ACC_e(x_i)\ne E_i(x_i)]$.

\item
Ask $(\exists x_1,\ldots,x_t)(\forall 1\le i\le t)[ACC_e(x_i)\ne E_i(x_i)].$
(Note that $ACC_e$ is decidable so this is a $(\exists)$ question.)
If YES then output YES.
If NO then output NO.
\end{enumerate}

\bigskip

\noindent
3) There exists a bounding function $f\TLE INF$ for $({\cal M},{\cal M'})$.

In the algorithm below we freely use Fact~\ref{fa:turing}.3 to phrase
$(\exists)(\forall)$-questions as queries to $INF$.

\noindent
{\bf Algorithm for $f$}

\begin{enumerate}
\item
Input$(n)$
\item
MAX=0.
\item
For every $\cal M'$-device $P$ of size $\le n$ do the following
\begin{enumerate}
\item
Ask $(\exists {\cal M}\hbox{-device $D$})(\forall x)[P(x)=D(x)]\hbox{?}$

\item
If YES then for $i=1,2,3,\ldots$ ask 
$(\exists {\cal M}\hbox{-device $D$}, |D|=i)(\forall x)[P(x)=D(x)]\hbox{?}$

\noindent
until the answer is YES. 
\item
Let $i$ be the value of $i$ when the last step stopped.
Note that $(\exists D, |D|=i)(\forall x)[P(x)=D(x)].$
If $i>MAX$ then $MAX=i$.
\end{enumerate}
\item
Output MAX.
\end{enumerate}
\end{proof}

\begin{corollary}\label{co:main}~
\begin{enumerate}
\item
If $f$ is a bounding function for (DPDA,PDA) then $INF\TLE f$.
\item
There exists a bounding function for (DPDA,PDA) such that $f\TLE INF$.
\item
If $INF\not\TLE f$ then for infinitely many $n$ there exists a 
language $A_n\in L(\DPDA)$
such that 
(1) any DPDA that recognizes $A_n$ requires size $\ge f(n)$ for $A_n$, but
(2) there is a PDA of size $\le n$ that recognizes $A_n$.
\item
If $f$ is a bounding function for (PDA,LBA) then $INF\TLE f$.
\item
There exists a bounding function for (PDA,LBA) such that $f\TLE INF$.
\item
If $INF\not\TLE f$ then for infinitely many $n$ there exists a 
language $A_n\in L(\PDA)$
such that 
(1) any PDA that recognizes $A_n$ requires size $\ge f(n)$,
(2) there is an LBA of size $\le n$ that recognizes $A_n$.
\end{enumerate}
\end{corollary}

\begin{proof}
We can apply Theorem~\ref{th:bdd} to all the relevant pairs
since all of the premises needed are either obvious or well known.
\end{proof}

\begin{note}\label{de:mult}
Since deterministic time classes are effectively closed under complementation
we can also apply Theorem~\ref{th:bdd} to get a corollaries about
any deterministic time class that contains $L(\PDA)$.
Let 
$$
\omega=
\inf
\{\alpha\st
\hbox{Two $n\times n$ Boolean matrices can be multiplied in time $O(n^\alpha)$}
\}.
$$
Le Gall~\cite{matmult} has the current best upper bound: $\omega< 2.3728639$.
We abuse notation by letting, for all $\alpha>0$,  $\DTIME(n^\alpha)$ 
be the set of all deterministic Turing
machines that run in time $O(n^\alpha)$.
Valiant~\cite{cflmatrix} showed that that, for all $\alpha>\omega$,
$L(\PDA) \subseteq L(\DTIME(n^{\alpha}))$.
If Boolean matrix multiplication really is in $\DTIME(n^\omega)$ then so is $L(\PDA)$.
(Lee~\cite{cflmatrixreq} showed that 
if $L(\PDA)\subseteq \DTIME(n^{3-\epsilon})$ then
$\omega \le 3-(\epsilon/3)$; 
therefore the problems of $L(\PDA)$ recognition and
matrix multiplication are closely linked.)
Hence, for all $\alpha>\omega$ (and possibly for $\omega$ also) we can obtain
a corollary about $\DTIME(n^\alpha)$ 
that is similar to Corollary~\ref{co:main}.
\end{note}

\section{c-Bounding Functions for PDAs}\label{se:cbd}

\begin{theorem}\label{th:cbd}~
\begin{enumerate}
\item
If $f$ is a c-bounding function for PDAs then $HALT\TLE f$.
\item
If $f$ is a c-bounding function for PDAs then $INF\TLE f$.
\item
There exists a c-bounding function $f\TLE INF$ for PDAs.
(This is almost identical to the proof of Theorem~\ref{th:bdd}.3
so we do not prove it.)
\item
If $INF\not\TLE f$ then for infinitely many $n$ 
there exists a language $A_n$
such that (1) $A_n,\overline{A_n}\in L(\PDA)$,
(2) there is no PDA of size $\le f(n)$ for $\overline{A_n}$, but
(3) there is a PDA of size $\le n$ for $A_n$.
(This follows from Part 2 so we do not prove it.)
\end{enumerate}
\end{theorem}

\begin{proof}

\begin{itemize}
\item
$P_1,P_2,\ldots,$ is a size-enumeration of PDAs.
\item
$f$ is a c-bounding function for PDAs.
\item
$g$ (when on two variables) 
is the computable function such that $\overline{ACC_{e,x}}$ is recognized
by PDA $P_{g(e,x)}$.
\item
$g$ (when on one variable) 
is the computable function such that $\overline{ACC_{e}}$ is recognized
by PDA $P_{g(e)}$.
\end{itemize}

\noindent
1) Let $t=f(g(e,x))$. 

$(e,x)\in HALT$ iff $(\exists 1\le a\le t)[L(P_a)\hbox{ is an accepting computation of $M_e(x)$}].$

Since both the non-emptiness problem and the membership problem for PDAs
is decidable this condition can be checked.

\bigskip

\noindent
2) Let $t=f(g(e))$.

$e\in INF$ 
$\implies ACC_e\notin L(\PDA)$ 
$\implies$ 
${ACC_e}\notin \{L(P_1),\ldots,L(P_t)\}$ 
$\implies$ 

$(\exists x_1,\ldots,x_t)(\forall 1\le i\le t)[P_i(x_i)\ne ACC_e(x_i)]$.

\bigskip

$e\notin INF$ 
$\implies ACC_e\hbox{ is finite }$
$\implies$
${ACC_e}\notin \{L(P_1),\ldots,L(P_t)\}$ 
$\implies$

$\neg(\exists x_1,\ldots,x_t)(\forall 1\le i\le t)[P_i(x_i)\ne ACC_e(x_i)]$.

We can now use $f\TLE HALT$ to determine if
$(\exists x_1,\ldots,x_t)(\forall 1\le i\le t)[P_i(x_i)\ne ACC_e(x_i)]$.
is true or not.
\end{proof}

\section{i-Bounding Functions for PDAs}\label{se:ibd}

\begin{definition}
We use the same conventions 
for Turing machines as in Definition~\ref{de:acc}.
Let $e,x\in\nat$.
\begin{enumerate}
\item
$ODDACC_{e,x}$ be the set of all sequences of config's represented by

$$
\$C_1
\$C_2^R
\$C_3
\$C_4^R
\$
\cdots
\$
C_s^R
\$
$$

such that
\begin{itemize}
\item
$|C_1|=|C_2|$ and $|C_3|=|C_4|$ and $\ldots$ and $|C_{s-1}| = |C_s|$.
\item
For all odd $i$, $C_{i+1}$ is the next config after $C_i$.
(We have no restriction on, say, how  $C_2$ and $C_3$ relate.
They could even be of different lengths.)
\item
$C_s$ represents an accepting config.
\end{itemize}
\item
Let $ODDACC_e= \bigcup_{x\in\nat} ODDACC_{e,x}$.
\item
$EVENACC_{e,x}$ be the set of all sequences of config's represented by

$$
\$C_1
\$C_2^R
\$C_3
\$C_4^R
\$
\cdots
\$
C_s^R
\$
$$

such that
\begin{itemize}
\item
$|C_2|=|C_3|$ and $|C_4|=|C_5|$ and $\ldots$ and $|C_{s-2}| = |C_{s-1}|$.
\item
For all even $i$, $C_{i+1}$ is the next config after $C_i$.
(We have no restriction on, say, how $C_3$ and $C_4$ relate. They could even
be of different lengths. We also have no restriction on $C_1$ except
that it be a config.)
\end{itemize}
\item
Let $EVENACC_e= \bigcup_{x\in\nat} EVENACC_{e,x}$.
\end{enumerate}
\end{definition}

Note that
\begin{enumerate}
\item
$(e,x)\in HALT$ iff $ODDACC_{e,x}\cap EVENACC_{e,x}$ contains only one
string and that string is an accepting computation of $M_e(x)$.
\item
$e\in INF$ iff $ODDACC_e\cap EVENACC_e \notin L(\PDA)$.
\end{enumerate}

Using these two facts you can prove the theorem below in a manner
similar to the proof of Theorem~\ref{th:cbd}.

\begin{theorem}~
\begin{enumerate}
\item
If $f$ is an i-bounding function for PDAs then $HALT\TLE f$.
\item
If $f$ is an i-bounding function for PDAs then $INF\TLE f$.
\item
There exists an i-bounding function $f\TLE INF$ for PDA.
\item
If $INF\not\TLE f$ then for infinitely many $n$ there exists 
languages $A_{n,1}$
and $A_{n,2}$ 
such that (1) $A_{n,1},A_{n_2}\in L(\PDA)$,
(2) there is no PDA of size $\le f(n)$ for $A_{n,1}\cap A_{n,2}$, but
(3) there is a PDA of size $\le n$ for $A_{n,1}\cap A_{n,2}$.
\end{enumerate}
\end{theorem}

\section{A Double-Exp For-Almost-All Result Via a Natural Language for (DPDA,PDA)}\label{se:bigpda}

We show that for almost all $n$ there is a (natural) language $A_n$ such that
$A_n$ has a small PDA but $\overline{A_n}$ requires a large PDA.
We then use this to show that for almost all $n$ there is a language
$A_n$ that has a small PDA but requires a large DPDA.
Neither of these results is new. Harel and Hirst~\cite{hhpda} have essentially proved
everything in this section. We include this section because some of our proofs are different
from theirs 
and because in most cases they do not explicitly state the theorems.
After every statement and proof in this section we briefly discuss what they did.
We denote their paper by HH.

\begin{lemma}\label{le:log}
Let $X,Y,Z$ be nonterminals.
Let $\Sigma$ be a finite alphabet.
\begin{enumerate}
\item
For all $n\ge 2$ there is a PDA of size $O(\log n)$ that generates $\{ Y^n \}$.
\item
For all $n\ge 2$ there is a PDA of size $O(\log n)$ that generates $\{a,b\}^n$.
\item
For all $n\ge 2$ there is a PDA of size $O(\log n)$ that generates $\{ Y^{\le n} \}$.
\item
For all $n\ge 2$ there is a PDA of size $O(\log n)$ that generates $\{a,b\}^{\le n}$.
\item
For all $n\ge 2$ there is a PDA of size $O(\log n)$ that generates $\{ Y^{\ge n} \}$.
\item
For all $n\ge 2$ there is a PDA of size $O(\log n)$ that generates $\{a,b\}^{\ge n}$.
\end{enumerate}
\end{lemma}

\begin{proof}

\noindent
1) We show that there is a CFG of size $\le 2\lg n$ that generates
$\{Y^n\}$ by induction on $n$.

If $n=2$ then the CFG for $\{YY\}$ is 

\noindent
$S\goes YY$ 

\noindent
which has $2=2\lg 2$ nonterminals.

\vfill\eject

If $n=3$ then the CFG for $\{YYY\}$ is

\noindent
$S\goes Y_1Y$

\noindent
$Y_1\goes YY$ 

\noindent
which has $3\le 2\lg 3$ nonterminals.

Assume that for all $m<n$ there is a CFG of size $\le 2\lg m$ for $\{Y^m\}$.
We prove this for $n$.

\begin{itemize}
\item
$n$ is even. Let $G'$ be the CFG for $\{Y^{n/2}\}$ with the start symbol replaced by $S'$.
The CFG $G$ for $\{Y^n\}$ is the union of $G'$ and the one rule $S\goes S'S'$.
This CFG has 
one more 
nonterminal than $G'$. Hence the number of nonterminals 
in $G$ is $\le 2\lg(n/2) + 1 \le 2\lg n$
\item
$n$ is odd. Let $G'$ be the CFG for $\{Y^{(n-1)/2}\}$ with the start symbol replaced by $S'$.
The CFG $G$ for $\{Y^n\}$ is the union of $G'$ and the two rules $S\goes YS''$ and $S''\goes S'S'$.
This CFG has two more nonterminals than $G'$. 
Hence the number of nonterminals in $G$ is
$\le 2\lg((n-1)/2) + 2 \le 2\lg n.$
\end{itemize}

\noindent
2) Add the the productions $Y\goes a$ and $Y\goes b$ to the CFG from Part 1.

\noindent
3) Add the production $Y\goes\epsilon$ to the CFG from Part 1.

\noindent
4) Add the production $Y\goes\epsilon$ to the CFG from Part 2.

\noindent
5) Let $G$ be the $O(\log n)$ sized CFG for $\{Y^n\}$ from Part 1. 
Let $G'$ be the $O(1)$ sized CFG for $Y^\kstar$.
The CFG for the concatenation of $L(G)$ and $L(G')$ is an $O(\log n)$ sized CFG for $\{Y^{\ge n}\}$.

\noindent
6) Add the productions $Y\goes a$ and $Y\goes b$ to the CFG from Part 5.
\end{proof}

\begin{note}
Lemma~\ref{le:log}.1 follows from the first paragraph of the 
proof of Proposition 16 in the journal version of HH
(Proposition 14 in the conference version).
They used PDAs (which they call $\es$-PDAs) where as we use CFGs and then convert them to PDAs.
\end{note}

\begin{theorem}\label{th:wwbar}
For almost all $n$ there exists a (natural) language $A_n$ such that the following hold.
\begin{enumerate}
\item
$A_n,\overline{A_n}\in L(\PDA)$.
\item
Any DPDA that recognizes $\overline{A_n}$ requires size $\ge 2^{2^{\Omega(n)}}$.
\item
There is a PDA of size $O(n)$ that recognizes $A_n$.
\end{enumerate}
\end{theorem}

\begin{proof}
We show there is a language $A_n$ such that
(1) $A_n,\overline{A_n}\in L(\PDA)$,
(2) any PDA that recognizes $\overline{A_n}$ requires size $\ge 2^{\Omega(n)}$, 
(3) there is a PDA of size $O(\log n)$ that recognizes $A_n$.
Rescaling this result yields the theorem. 

Let $W_n = {\{ ww \st |w|=n\}}$.
Let $A_n = \overline{W_n}$.

\noindent
1) $A_n$ is cofinite, so both $A_n$ and $\overline{A_n}$ are in $L(\PDA)$.

\noindent
2) Filmus~\cite{lbcfg} showed that
any CFG for $W_n$ requires size $\ge 2^{\Omega(n)}$.  
Hence by Example~\ref{ex:bd}.4 any PDA for $W_n=\overline{A_n}$
requires size $\ge 2^{\Omega(n)}$.

\bigskip

\noindent
3) We present a CFG for $A_n$ of size $O(\log n)$. 
By Example~\ref{ex:bd}.5 this suffices to obtain a PDA of size $O(\log n)$.
We will freely use that Lemma~\ref{le:log} yields CFG's of size $O(\log n)$ by Example~\ref{ex:bd}.4.

Note that if $x\in A_n$ then either $|x|\le 2n-1$, $|x|\ge 2n+1$, 
or there are two letters in $x$
that are different and are exactly $n-1$ apart.
These sets are not disjoint.

The CFG is the union of three CFGs. 
The first one generates all strings of length $\le 2n-1$.
By Lemma~\ref{le:log} there is such a CFG of size $O(\log n)$.
The second one generates all strings of length $\ge 2n+1$.
By Lemma~\ref{le:log} there is such a CFG of size $O(\log n)$.

The third one generates all strings of length $\ge 2n$ where there are two letters that are
different and exactly $n-1$ apart (some of these strings are also generated by the second
CFG).
By Lemma~\ref{le:log} there is a CFG
$G'$ of size $O(\log n)$ that  generates all strings of length $n-1$.
Let $S'$ be its start symbol.
$G'$ will be part of our CFG $G$.

Our CFG has start symbol $S$, all of the rules in $G'$, and the following:

$S\goes UaS'bU \quad | \quad UbS'aU$

$U\goes aU \quad | \quad bU \quad | \quad \epsilon$

The union of the three CFG's clearly yields a CFG of size $O(\log n)$ for $A_n$.
\end{proof}

\begin{note}
Theorem~\ref{th:wwbar} is implicit in Proposition 27 of HH. They use
the language 
$$\Sigma^\kstar - \{w\$w\$w\$w \st |w|=n \}.$$
\end{note}

We can now obtain a double exponential result about (DPDA,PDA).

\begin{theorem}\label{th:wwpda}
For almost all $n$ there exists a (natural) language $A_n$ such that  the following hold.
\begin{enumerate}
\item
Any DPDA that recognizes $A_n$ requires size $\ge 2^{2^{\Omega(n)}}$.
\item
There is a PDA of size $O(n)$ that recognizes $A_n$.
\end{enumerate}
\end{theorem}

\begin{proof}
Let $A_n$ be as in Theorem~\ref{th:wwbar} (note that it is scaled).
We already have that $A_n$ has a PDA of size $O(n)$.
We show that any DPDA for $A_n$ is of size $\ge 2^{2^{\Omega(n)}}$.
Let $P$ be an DPDA for $A_n$. 
By Example~\ref{ex:cbd}.2 there is a DPDA $P'$ for $\overline{A_n}$ of size
$O(|P|)$.  By Theorem~\ref{th:wwbar} $|P'|\ge 2^{2^{\Omega(n))}}$,
hence $|P|\ge 2^{2^{\Omega(n)}}$,
\end{proof}

\begin{note}
Theorem~\ref{th:wwpda} is a special case of Corollary 30 of HH.
\end{note}


\section{A Double-Exp For-Almost-All Result Via a Natural Language for (PDA,LBA)}\label{se:bigishlba}

We show that for almost all $n$ there is a (natural) language
$A_n$ that has a small LBA but requires a large PDA.

\begin{theorem}\label{th:wwlba}
For almost all $n$ there exists a (natural) language $A_n$ such that  the following hold.
\begin{enumerate}
\item
Any PDA that recognizes $A_n$ requires size $\ge 2^{2^{\Omega(n)}}$.
\item
There is an LBA of size $O(n)$ that recognizes $A_n$.
\end{enumerate}
\end{theorem}

\begin{proof}
We show there is a language $A_n$ such that
(1) any PDA for $A_n$ requires size $\ge 2^{\Omega(n)}$ and
(2) there is an LBA of size $O(\log n)$ for $A_n$.
Rescaling this result yields the theorem. Let

$$A_n = \{ w\$w \st |w|=n\}.$$

\noindent
1) Filmus~\cite{lbcfg} showed that
any CFG for $A_n$ requires size $\ge 2^{\Omega(n)}$.  
Hence, by Example~\ref{ex:bd}.4, any PDA for $A_n$
requires size $\ge 2^{\Omega(n)}$.

\bigskip

\noindent
2) We present a CSG for $W_n$ of size $O(\log n)$.
By Example~\ref{ex:bd}.7 this yields an LBA of size $O(\log n)$.

Here is the CSG for $\{ w\$w \st |w|=n\}$. 

$S\goes Y^nW$ (actually use the CFG from Lemma~\ref{le:log} of size $O(\log n)$ to achieve this)

$Y\goes aA \quad | \quad bB$

$Aa \goes aA$

$Ab \goes bA$

$Ba \goes aB$

$Bb \goes bB$

$AW \goes Wa$

$BW \goes Wb$

$W \goes \$ $
\end{proof}

\section{A Ginormous For-Almost-All Result for (PDA,LBA)}\label{se:biglba}

Meyer and Fisher~\cite{ecodesc} say the following in their {\bf Further Results} Section:

\centerline{{\it $\ldots$ context-sensitive grammars may be arbitrarily more succinct than context-free grammars $\ldots$ }}

The reference given was a paper of Meyer~\cite{meyerpl}. 
That paper only refers to Turing Machines.
We exchanged emails with Meyer about this and he informed us that
his techniques could be used to obtain the result that is
Theorem~\ref{th:cslvvsmall} below.
Rather than work through his proof we provide our own.
Our proof is likely similar to his; however, we use the
closure of $L(\LBA)$ under complementation~\cite{Immerman,Sz}
which was not available
to him at the time.

We assume that all LBAs are modified so that, on input $x$, if a branch does not 
terminate in time $2^{|x|^2}$ then that branch will halt and reject. Hence
every branch either halts and accepts or halts and rejects.

Let $P_1,P_2,\ldots$ be a size-enumeration of all PDAs.
We assume that $P_e$ is of size $\ge e$.
We also have a list $N_1,N_2,\ldots,$ of LBAs such
that $L(N_i)=L(P_i)$ and 
(by the effective closure of $L(\LBA)$ under complementation)
$N_1',N_2',\ldots$
such that $L(N_i') = \overline{L(P_i)}$.
Note that $N_1,N_2,\cdots$ {\it is not} a list of all LBAs.

The following will be key later:
Let $x\in\Sigma^\kstar$ and imagine
running $N_i(x)$, for each path
noting if it said Y or N, and then running
$N_i'(x)$, and then noting if that path
said Y or N. So each path ends up
with a NN, NY, YN, or YY.
\begin{itemize}
\item
$x\in L(P_i)$:
some path says YN,
some paths might say NN,
but no path says NY or YY.
\item
$x\notin L(P_i)$:
some path says NY,
some path might says NN,
but no path says YN or NN.
\end{itemize}

Note that $N_i$ and $N_i'$ run in $O(|x|)$ space.

\begin{theorem}\label{th:cslvvsmall}
Let $f\TLE HALT$.
For almost all $n$ there exists a finite set 
$A_n$ such that the following hold.
\begin{enumerate}
\item
Any PDA that recognizes $A_n$ requires size $\ge f(n)$.
\item
There is an LBA of size $O(n)$ that recognizes $A_n$.
\end{enumerate}
\end{theorem}

\begin{proof}
We construct the language $A_n$ by describing an LBA for it (really an $\NSPACE(|x|)$ algorithm).
The idea is that $A_n$ will
be diagonalized against all small PDAs.
The algorithm will run in $O(|x|)$ space. We will comment on the constant in the
$O(|x|)$ later.

Since $f\TLE HALT$, by Fact~\ref{fa:turing}.6,  
there exists a computable $g$ such that
$(\forall n)[f(n)=\lim_{s\goes\infinity} g(n,s)]$.
We can assume $g(n,s)$ can be computed in space $O(\log(n+s))$.

Fix $n$. We describe the algorithm for $A_n$. The set we construct
will satisfy the following requirements:

For $1\le i\le f(n)$ (which we do not know)

$R_i: A_n \ne L(P_i)$.

This is only a finite number of requirements; however,
we do not know $f(n)$. We will get around this by approximating
$f(n$) via $g(n,s)$.

The set $A_n$ will be a subset of $a^\kstar$.

\noindent
{\bf ALGORITHM for $A_n$}
\begin{enumerate}
\item
Input($a^s$).
\item
Compute $t=g(n,s)$.
\item
Deterministically simulate 
$A_n$ on the strings $\{\epsilon, a, a^2, \ldots, a^{\lg^* s}\}$.
Do not store what the results are; however, store which requirements
indexed $\le t$ are satisfied. If so many were satisfied that you can't store them
in space $\le \log s$
then reject and halt.
\item
If all of the requirements $P_i$ as $1\le i\le t$ are satisfied
then reject and halt.
\item
(Otherwise)
Let $i$ be the least elements of $\{1,\ldots,t\}$ such
that $R_i$ has not been seen to be satisfied. 
Run (nondeterministically) $N_i(x)$ and then $N_i'(x)$.
Any path that yields NN outputs NO.
There will be no paths that yields YY.
Any path that yields NY output YES (this is diagonalization--- a
NY means that $x\notin L(P_i)$).
Any path that yields YN output NO (this is diagonalization--- a
YN means that $x\in L(P_i)$).
Requirements $R_i$ is satisfied.
\end{enumerate}
{\bf END OF ALGORITHM for $A_n$}

By the definition of $g$ there exists 
$s_0$ such that, for all $s,s'\ge s_0$, $g(n,s)=g(n,s_0)=f(n)$.

We show, by induction on $i$,  that for all $i\le f(n)$, $R_i$ is satisfied.
Assume that for all $i'<i\le f(n)$, $R_{i'}$ is satisfied.
Let $s_1>s_0$ be the least $s$ such that all $R_{i'}$ 
with $i'<i$ are satisfied after $A_n(a^s)$ is determined.
Let $s_2>s_1$ be the least $s$ such that for all inputs $a^{\ge s}$ 
the algorithm 
deterministically simulates $A_{s_1}$ and hence 
notices that, for all $i'<i$, $R_{i'}$ is satisfied.
If $R_i$ is satisfied on some input in $a^{\le s_2}$ then we are done.
Otherwise note that on input $a^{s_2}$ the algorithm will
notice that $R_i$ is not satisfied and satisfy it.

How big is the LBA for $A_n$?
The LBA only needs the parameter $n$ and a constant number of
instructions. Hence their is an LBA of size $O(\log n)$; however,
we only need that there is an LBA of size $O(n)$.

We show that the algorithm for $A_n$ is in $\NSPACE(O(|x|)$.
Let $g_{\max} = \max\{g(n,s) \st s\in \nat \}$. 
Since $\lim_{s\goes\infinity} g(n,s)$ exists 
$g_{\max}$ is well defined.
It depends on $n$ but not on the input; 
hence $g_{\max}$ is a constant.
The first four steps of the algorithm take 
$\le \lg^*(|x|)+g_{\max}$ space to execute.
For large $|x|$ this is far less than $|x|$.
Step 5 is the only nondeterministic step. Each branch is the result of running 
a branch of the $\NSPACE(O(|x|)$ machines 
$N_i(x)$ and $N_i'(x)$ where $1\le i \le g_{\max}$.
Hence there is a
constant $c$ such that for all $x$ each branch of the computation takes $\le c|x|$ 
space. Therefore the algorithm for $A_n$ is in $\NSPACE(O(|x|)$.
\end{proof}

\section{Open Problems}

We have pinned down the exact Turing degree of the bounding function
for (DPDA,PDA) and (PDA,LBA). The exact Turing degree for the bounding
functions for (DPDA,UCFG) and (UCFG,PDA) are open.

We have obtained natural languages that show the (1) bounding function
for (DPDA,PDA) and (PDA,LBA), and (2) the c-bounding function for PDAs,
are at least double exponential. It is open to
find natural languages that show a larger lower bounds.

\section{Acknowledgment}

We thank Albert Meyer for conversations that helped clarify the history of 
Theorem~\ref{th:cslvvsmall}.
We thank 
Hermann Gruber, 
Albert Meyer,
Leslie Valiant, 
and the referees for
some references we had overlooked.
We thank Yuval Filmus, Karthik Gopalan, Rebecca Kruskal, Sam Zbarsky, 
and the referees for corrections and proofreading.
We thank Yuval Filmus and one of the referees for pointing out an error
in the original proof of Theorem~\ref{th:wwlba} and supplying us with a way to
fix the proof.
We thank Keith Ellul, Bryan Krawetz, Jeffrey Shallit, and Ming-wei Wang
whose paper~\cite{regexp} inspired this paper; and we thank
Jeffrey Shallit who brought this paper to our attention.


\end{document}